\pdfoutput=1                            
\def\lcversion{}					    

\documentclass[nonote,noversion,squeezed]{cdcarticle}

\def\MODE{2}

\usepackage{standalone}
\usepackage{graphicx}
\graphicspath{{./graphics/}}

\usepackage{cite}                   
\usepackage{math}                   
\usepackage[utf8]{inputenc}			
\usepackage[english]{babel}			
\usepackage[hidelinks]{hyperref}    

\usepackage{enumitem}               
\usepackage{bbm}                    
\usepackage{comment}

\newcommand{\1}{\mathbf{1}}
\newcommand{\0}{\mathbf{0}}

\newcommand{\Pp}{\mathcal{P}}
\newcommand{\Tc}{\mathcal{U}}
\newcommand{\U}{\mathcal{V}}
\newcommand{\K}{\mathcal{K}}
\newcommand{\V}{\mathcal{W}}	
\newcommand{\G}{\mathcal{G}}	

\newcommand{\sbmat}[1]{\left(\begin{smallmatrix}#1\end{smallmatrix}\right)}
\newcommand{\squeezemat}[1]{\addtolength{\arraycolsep}{-#1}}

\newcommand\blfootnote[1]{%
	\begingroup
	\renewcommand\thefootnote{}\footnote{#1}%
	\addtocounter{footnote}{-1}%
	\endgroup
}

\usepackage{subcaption}
\usepackage[margin=3mm,font=small,labelfont=bf]{caption}

\begin{document}
	
\title{Explicit Agent-Level Optimal Cooperative Controllers for
	Dynamically Decoupled Systems with Output Feedback}

\if\MODE1
\author{Mruganka Kashyap \and Laurent Lessard}
\else
\author{Mruganka Kashyap \and Laurent Lessard}\fi

\note{Submitted to CDC 2019}
\maketitle



\begin{abstract}
	We consider a dynamically decoupled network of agents each with a local output-feedback controller. We assume each agent is a node in a directed acyclic graph and the controllers share information along the edges in order to cooperatively optimize a global objective. We develop explicit state-space formulations for the jointly optimal networked controllers that highlight the role of the graph structure. Specifically, we provide generically minimal agent-level implementations of the local controllers along with intuitive interpretations of their states and the information that should be transmitted between controllers.
\end{abstract}

\if\MODE1\else

\blfootnote{M.~Kashyap and L.~Lessard are with the Department of Electrical and Computer Engineering, at the University of Wisconsin--Madison, Madison, WI~53706, USA. L.~Lessard is also with the Wisconsin Institute for Discovery at the University of Wisconsin--Madison.
	\texttt{\{mkashyap2,laurent.lessard\}@wisc.edu}\\[1mm]
	This material is based upon work supported by the National Science Foundation under Grant No.~1710892.}
\fi

\section{Introduction}\label{sec:intro}

Dynamically decoupled systems are examples of multi-agent systems where each agent has independent dynamics but the agents' controllers share information via a communication network to achieve a common goal.
Dynamically decoupled architectures occur naturally in the context of swarms of unmanned aerial vehicles (UAVs). For example, a swarm of UAVs might be deployed to survey an uncharted region or to optimize geographic coverage to combat a forest fire.

\begin{figure}[ht]
	\centering
	\includegraphics{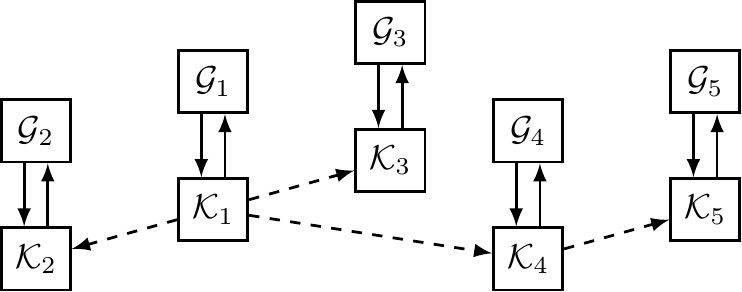}
	\caption{Example of a dynamically decoupled system. The plants $\mathcal{G}_i$ each have controllers $\mathcal{K}_i$ that share information instantaneously along a directed acyclic graph. The edge $\mathcal{K}_1 \rightarrow \mathcal{K}_5$ is implied but not shown.}
	\label{fig:cartoon}
\end{figure}

\noindent In this paper, we abstract the communication network as a directed acyclic graph (DAG). An example with five nodes is shown in Figure~\ref{fig:cartoon}. Each plant $\mathcal{G}_i$ is assumed to obey state-space linear dynamics of the form
\begin{subequations}\label{eq:dynamic_ss}
	\begin{align}
		\dot{x}_i &= A_{ii}x_i + B_{1_{ii}}w_i + B_{2_{ii}}u_i\\
		y_i &= C_{2_{ii}}x_i + D_{{21}_{ii}}w_i,
	\end{align}
\end{subequations}
where $x_i(t)\in\R^{n_i}$, $u_i(t)\in\R^{m_i}$, and $y_i(t)\in\R^{p_i}$ are the state, input, and measurement for each agent, and $w_i(t)\in\R^{q_i}$ is exogenous standard Gaussian noise, assumed to be independent across agents and time.

The global cooperative task is to optimize a standard infinite-horizon quadratic objective that couples the states and inputs of all agents. Specifically, if we form the global state by aggregating agents $\{1,\dots,N\}$ as $x^\tp \defeq \bmat{x_1^\tp & \cdots & x_N^\tp}$ and similarly for the aggregated $u$, $y$, and $w$, the cost function is of the form
\begin{equation}\label{eq:cost_function}
\J = \lim_{T\to\infty} \frac{1}{T}
\ee \int_{0}^T \normm{ C_1 x(t) + D_{12} u(t) }^2 \,\mathrm{d}t,
\end{equation}
where the expectation $\ee$ is taken with respect to the exogenous noise $w$. If each controller $\mathcal{K}_i$ could measure the past history of all global measurements $y$, then finding the $\mathcal{K}_i$ that optimizes $\mathcal{J}$ in~\eqref{eq:cost_function} would be a standard Linear Quadratic Gaussian (LQG) problem~\cite{ZDG}.

What makes our problem more difficult is that each $\mathcal{K}_i$ only has access to the measurement histories of its \textit{ancestor} nodes in the DAG. In the example of Figure~\ref{fig:cartoon}, $\mathcal{K}_1$ only measures $y_1$, $\mathcal{K}_3$ measures $\{y_1,y_3\}$, $\mathcal{K}_5$ measures $\{y_1,y_4,y_5\}$, etc. All measurements are instantaneous.

\paragraph{Contributions.} We present an \textit{agent-level} solution to the dynamically decoupled output-feedback problem. Specifically, we provide state-space realizations for the individual $\mathcal{K}_i$. Moreover, we provide intuitive interpretations for the controller states that highlight the role of the underlying graph structure and characterize the signals that should be transmitted between controllers.
This work builds on~\cite{kim2012optimal,kim2015explicit}, which reported a state-space realization of the aggregated optimal controller but no agent-level details.
Our results are presented in continuous time with an infinite-horizon cost but can easily be generalized to discrete time and/or a finite-horizon cost.

\paragraph{Literature review.}
The elegant solutions found for optimal centralized control do not generally carry over to the decentralized case. It was shown for example that under LQG assumptions, linear compensators can be strictly suboptimal~\cite{witsenhausen1968counterexample}. Several related decentralized control problems were shown to be intractable~\cite{blondel}.

Tractability is closely related to how information is shared between the different controllers. For LQG problems with \textit{partially nested} information, the optimal controller is linear \cite{ho1972team}. When the control structure is \textit{quadratically invariant} with respect to the plant, the problem of finding the optimal linear controller can be cast as a convex optimization over a set of stable Youla parameters~\cite{rotkowitz2002decentralized, rotkowitz2005characterization}. 
Under mild assumptions, convexity implies the existence of a unique optimal controller and efficient numerical methods for finding it. Numerical approaches include vectorization~\cite{rotkowitz_vectorization,vamsi_elia} and Galerkin-style finite-dimensional approximation~\cite{scherer02,voulgaris_stabilization}.

The optimal decentralized controller often carries a rich structure, similar to the estimator-controller separation structure of classical centralized LQG control~\cite{wonham1968separation}. Such structure is lost in the aforementioned numerical approaches, which prompted recent efforts to solve these problems in an \textit{explicit} way that uncovers the structure and intuition reminiscent of the centralized theory.

Explicit solutions for problems that are both partially nested and quadratically invariant have appeared for a variety of special cases. For example, state-feedback problems over general graphs where agent dynamics are coupled according to the graph structure~\cite{shah2013cal,spdel,swigart_spectral}. Some output-feedback structures have also been solved, including two-player~\cite{lessard2015optimal}, broadcast~\cite{lessard2012decentralized}, triangular~\cite{tanaka2014optimal}, and dynamically decoupled~\cite{kim2012optimal,kim2015explicit} cases. In the aforementioned works, explicit solutions typically consist of a state-space realization for the structured map $y \mapsto u$ between the global measurements $y$ and the global inputs $u$.

In this paper, we look at the dynamically decoupled case and drill down even further. We provide an \textit{agent-level} solution, which is an explicit state-space realization of the optimal controllers for each of the agents that is both minimal and easily interpretable.
Section~\ref{sec:notation} covers notation, formal assumptions, and other preliminaries. Section~\ref{sec:result} presents our alternative form for the optimal controller, followed by interpretations of the controller states and signals in Section~\ref{sec:diss} and the explicit agent-level structure in Section~\ref{sec:impl}. We conclude in Section~\ref{sec:conclufuture}.  

\section{Preliminaries}\label{sec:notation}

\paragraph{Notation.} A square matrix is \textit{Hurwitz} if all eigenvalues have a strictly negative real part. $\Rp$ is the set of real rational proper transfer matrices and $\RHinf$ is the set of real-rational proper transfer matrices with no poles on the imaginary axis and no unstable poles. Similarly, $\RHtwo$ is the set of real-rational strictly proper transfer matrices with no poles on the imaginary axis and no unstable poles. We use the state-space notation
\begin{align*}
	\G(s) = \left[\begin{array}{c|c}%
		A&B\\ \hlinet C&D\end{array}\right]= D+C(sI-A)^{-1}B.
\end{align*}
If the realization above is stabilizable and detectable, then $\mathcal{G} \in \RHinf$ if and only if $A$ is Hurwitz. The $\Htwo$ norm of $\mathcal{G}$ is $\norm{\mathcal{G}}_2^2 = \frac{1}{2\pi}\int_{-\infty}^{\infty} \trace(\mathcal{G}^*(jw)\mathcal{G}(jw))\,\mathrm{d}w$. All transfer matrices are represented using calligraphic symbols.

We use $N$ to denote the total number of agents and $[N] \defeq \{1,\dots,N\}$. The global state dimension is $n \defeq n_1+\cdots+n_N$ and similarly for $m$ and $p$. $I_k$ will represent an identity matrix of size $k$. $\blkdiag(\{X_i\})$ is the block-diagonal matrix formed by the blocks $\{X_1,\dots,X_n\}$ and $\diag(X)$ is the block-diagonal matrix formed by the diagonal blocks of $X$.
The zeros used throughout are matrix or vector zeros and their sizes are determined from context. The symbol $\star$ represents a possible nonzero value in a matrix with dimensions inferred from context.

For $i\in [N]$, we write $\underline{i} \subseteq [N]$ to denote the \textit{descendants} of node $i$, i.e., the set of nodes $j$ such that there is a directed path from $i$ to $j$. Likewise, we write $\bar{i} \subseteq [N]$ to denote the \textit{ancestors} of node $i$. Similarly, $\bar{\bar{i}}$ and $\underline{\underline{i}}$ denote the \textit{strict ancestors} and \textit{strict descendants}, respectively.
For example, in the graph of Figure~\ref{fig:cartoon}, $\underline{1} = [5]$, $\underline{4} = \{4,5\}$, and $\bar{\bar{5}} = \{1,4\}$.
We also use this notation as a means of indexing matrices. For example, if $X$ is a $N\times N$ block matrix, then we can write:
$X_{2\underline{4}} = \bmat{X_{24} & X_{25}}$.

Finally, we will require the use of specific partitions of the identity matrix. $I_{n} \defeq \blkdiag(\{I_{n_i}\})$ and for each agent $i \in [N]$, we define $E_{n_i} \defeq (I_{n})_{:i}$ (the $i^\textup{th}$ block column of $I_{n}$). Similar to the descendant and ancestor definitions, $n_{\underline{i}}=\sum_{k\in \underline{i}}n_k$ and $n_{\bar{i}}=\sum_{k\in \bar{i}}n_k$. The dimensions of $E_{{n_{\bar{i}}}}$ and $E_{{n_{\underline{i}}}}$ are determined by the context of use.  $1_{n}$ is the $n\times 1$ matrix of $1$'s. The definitions above also hold for other sets of indices, for example if we replace $n$ by $m$ or $p$.

\paragraph{Problem statement.}
Aggregating~\eqref{eq:dynamic_ss} across agents, we can form global dynamics where $x$, $u$, $w$, and $y$ are the global state, input, exogenous noise, and measured output, respectively. We also define the regulated output $z = C_1 x + D_{12}u$. The transfer matrix $(w,u)\mapsto (z,y)$ is
\begin{align*}
	\left[\begin{array}{c c}
		\Tc & \U\\
		\V & \G\end{array}\right] \defeq
	\left[\begin{array}{c|cc}
		A & B_1 & B_2 \\ \hlinet
		C_1 & 0 & D_{12} \\
		C_2 & D_{21} & 0\end{array}\right].
\end{align*}
The dynamically decoupled structure of~\eqref{eq:dynamic_ss} implies that 
$A \defeq \blkdiag(\{A_{ii}\})$ and similarly for $B_1$, $B_2$, $C_2$, $D_{21}$. The cost matrices $C_1$ and $D_{12}$ are not assumed to have any special structure. It follows that $\V$ and $\G$ are block diagonal but $\Tc$ and $\U$ need not be.

The problem is to find the optimal controller $\mathcal{K}$ that solves the optimization problem
\begin{equation}\label{opt}
\begin{aligned}
& \underset{\K \; \stabilizes \; \G}{\minimize}
& & \normm{\,\Tc + \U\K(I-\ \G\K)^{-1}\V\,}_2^2 \\
& \subject
& & \K \in \mathbb{S}_{\Rp}.
\end{aligned}
\end{equation}
We assume the DAG characterizing the information-sharing pattern has an adjacency matrix $S \in \{0,1\}^{N\times N}$, and $\mathbb{S}_{\Rp}$ is the set of proper matrix transfer functions with a block-sparsity pattern corresponding to $S$. Namely,
$\mathbb{S}_{\Rp} \defeq \set{ \K \in \Rp^{m\times p} }{ \K_{ij}=0 \text{ if } S_{ij}=0 }$. Note that the sub-blocks of $\K$ satisfy $\K_{ij} \in \Rp^{m_i \times p_j}$. The optimal controller in this case is known to be linear and finite-dimensional~\cite{kim2015explicit}. So there is no loss of generality in assuming $\K \in \mathbb{S}_{\Rp}$ in~\eqref{opt}.
In the sequel, we make several simplifying assumptions regarding the adjacency matrix and the system matrices.
\begin{assumption}[\textbf{Information sharing structure}]
	\label{Ass:InfoStruct}
	Information is shared instantaneously between agent $i$ and its descendants $\underline{i}$, so there is an implicit edge from $i$ to each $j\in\underline{i}$ (the graph is transitively closed). Thus, any directed cycle can be treated as a single node and the graph becomes a DAG. By relabeling the nodes according to the paritial ordering induced by the DAG, the adjacency matrix $S\in \R^{N\times N}$ becomes lower triangular. $S$ is also invertible because each node has a self-loop.
\end{assumption}

\noindent Assumption~\ref{Ass:InfoStruct} is critical in ensuring that the problem is partially nested~\cite{ho1972team} and quadratically invariant~\cite{rotkowitz2005characterization,rotkowitz_vectorization}, so the optimal controller is linear and finite-dimensional.

\paragraph{Riccati assumptions.} Matrices $(A,B,C,D)$ are said to satisfy the \emph{Riccati assumptions}~\cite{kim2015explicit} if:
\begin{enumerate}[label*=R\arabic*., ref=R\arabic*]
	\item $C^\tp D=0$ and $D^\tp D>0$. \label{harf}
	\item $(A,B)$ is stabilizable.
	\item $\begin{bmatrix}A-j\omega I&B\\C&D\end{bmatrix}$ is full column rank for all $\omega \in \R$.
\end{enumerate}
If the Riccati assumptions hold, there is a unique stabilizing solution to the associated algebraic Riccati equation, which we write as $(X,F)=\ric(A,B,C,D)$. Thus $X \succ 0$ satisfies
$
A^\tp X+XA+C^\tp C -XB(D^\tp D)^{-1}B^\tp X=0
$ with $A+BF$ Hurwitz, where $F\defeq -(D^\tp D)^{-1}B^\tp X$.

\begin{assumption}[\textbf{System assumptions}]
	\label{Ass:System}
	For the $N$ interacting agents, we will assume the following.
	\begin{enumerate}[label*=\arabic{thm}.\arabic*., ref=\arabic{thm}.\arabic*]
		\item As described in Section~\ref{sec:notation}, the system is dynamically decoupled, $A \defeq \blkdiag(\{A_{ii}\})$ and similarly for $B_1$, $B_2$, $C_2$, and $D_{21}$. The cost matrices $C_1$ and $D_{12}$ are not assumed to have any special structure. \label{Ass:System_dynamic}
		\item $A_{ii}$ is Hurwitz for all $i\in[N]$, i.e., $A$ is Hurwitz. \label{Ass:System_hurwitz}
		\item The Riccati assumptions hold for $(A,B_2,C_1,D_{12})$ and for $(A_{ii}^\tp ,C_{2_{ii}}^\tp ,B_{1_{ii}}^\tp ,D_{{21}_{ii}}^\tp )$ for all $i\in[N]$. \label{Ass:System_riccati}
	\end{enumerate}
\end{assumption}
\noindent Assumption~\ref{Ass:System_hurwitz} is an assumption of nominal stability, carried over from Kim and Lall \cite{kim2015explicit}. Assumption~\ref{Ass:System_riccati} ensures the necessary condition that the centralized LQR problem and the individual agents' estimation problems are non-degenerate.

\section{Main Results}\label{sec:result}
Kim and Lall \cite{kim2015explicit} provide an explicit state-space formula for the controller that optimizes~\eqref{opt} under Assumptions~\ref{Ass:InfoStruct}--\ref{Ass:System}. This will be the starting point for our work.
\begin{lem}[{\!\!\cite[Cor.~14]{kim2015explicit}}]\label{lem:1}
	Suppose Assumptions \ref{Ass:InfoStruct}--\ref{Ass:System} hold. 
	The solution to~\eqref{opt} is $\K = \Pp(I+\G \Pp-\diag(\G \Pp))^{-1}$, where the $i^\text{th}$ column of $\Pp\in \mathbb{S}_{\Rp}$ is given by
	\begin{align}\label{eq:Pdef}
		\Pp_{\underline{i}i}
		=\left[\begin{array}{c|c}%
			A_{\underline{ii}}+B_{2_{\underline{ii}}}F^i+\bmat{L^i\\0}C_{2_{i\underline{i}}} & \bmat{-L^i\\0} \\ \hlinet F^i & 0
		\end{array}\right]
	\end{align}
and we define the control and estimation gains as
\begin{subequations}\label{eq:ARE_eq}
	\begin{align}
		(X^i,F^i) &\defeq
		\ric(A_{\underline{ii}},B_{2_{\underline{ii}}},C_{1_{:\underline{i}}},D_{12_{:\underline{i}}}) \\
		(Y^i,{L^i}^\tp) &\defeq
		\ric(A_{ii}^\tp,C_{2_{ii}}^\tp,B_{1_{ii}}^\tp,D_{{21}_{ii}}^\tp).
	\end{align}
\end{subequations}
\end{lem}
Lemma~\ref{lem:1} provides a state-space realization of the optimal controller. Our first main result is an even more explicit formula for the optimal controller $\K$ that highlights the duality between estimation and control. We will also see in Section~\ref{sec:impl} that this new formula is more readily implementable on an agent-level basis.

\begin{thm}
	\label{thm:2}
	Suppose Assumptions \ref{Ass:InfoStruct}--\ref{Ass:System} hold. A realization for the optimal controller that solves~\eqref{opt} is
	\begin{align}\label{Kopt}
		\K=\left[\begin{array}{c|c}%
			\begin{matrix}
				\bar{A}+\bar{L}\bar{C}+\bar{B}\bar{S}_m\bar{F}\bar{S}_n^{-1}
			\end{matrix} &\begin{matrix}{-\bar{L}\bar{\1}_p}\end{matrix}\\\hlinet\begin{matrix}{\bar{\1}_m^\tp \bar{F}\bar{S}_n^{-1}}\end{matrix}&0\end{array}\right],
	\end{align}
	where we have defined the new symbols
	\begin{gather*}
		\bar{A} \defeq I_{N} \otimes A, \qquad
		\bar{B} \defeq I_N \otimes B_2, \qquad
		\bar{C} \defeq I_N \otimes C_2, \\
		\bar{S}_m \defeq S\otimes I_m, \qquad
		\bar{S}_n \defeq S\otimes I_{n}, \qquad
		\bar{S}_p \defeq S\otimes I_p,\\
		\bar{\1}_m \defeq 1_N\otimes I_m, \qquad
		\bar{\1}_p \defeq 1_N\otimes I_p.
	\end{gather*}
	For all $i\in [N]$, with $F^i$ and ${L^i}^{\tp}$ from~\eqref{eq:ARE_eq}, we also defined
	\begin{gather*}
		{L^{\bar{i}}} \defeq \blkdiag(\{L^j\}_{j \in \bar{i}})\\
		\bar{L} \defeq \blkdiag(\{L_i\})\quad\text{where: }
		L_i\defeq E_{n_{\bar{i}}}L^{{\bar{i}}}E_{p_{\bar{i}}}^\tp \\
		\bar{F} \defeq \blkdiag(\{F_i\})\quad\text{where: }
		F_i \defeq E_{m_{\underline{i}}}F^i E_{n_{\underline{i}}}^\tp.
	\end{gather*}
\end{thm}
\begin{proof}
	See Appendix \ref{sec:appendix} for a complete proof.
\end{proof}
\begin{rem}\label{rem:minimality}
	The realization~\eqref{Kopt} of Theorem~\ref{thm:2} is not minimal. Indeed, the gain matrices $L_i$ and $F_i$ have been padded with zeros so that they have compatible dimensions. We have $\bar A \in \R^{nN\times nN}$, where recall $N$ is the number of agents and $n \defeq n_1+\dots+n_N$ is the global aggregate state dimension of all agents. We use this non-minimal form to simplify algebraic manipulations; we present a generically minimal realization of~\eqref{Kopt} in Section \ref{sec:impl}.
\end{rem}
Applying the state transformation $x \mapsto \bar S_n x$ to~\eqref{Kopt}, we obtain a dual realization for the optimal controller.
\begin{cor}\label{corr}
	Suppose Assumptions \ref{Ass:InfoStruct}--\ref{Ass:System} hold. A realization for the optimal controller that solves~\eqref{opt} is
	\begin{align}\label{Kopt2}
		\K ={\left[\begin{array}{c|c}%
				\begin{matrix}\bar{A}+\bar{S}_n^{-1}\bar{L}\bar{S}_p\bar{C}+\bar{B}\bar{F}\end{matrix}&\begin{matrix}-\bar{S}_n^{-1}\bar{L}\bar{\1}_p\end{matrix}\\\hlinet {\begin{matrix}\bar{\1}_m^\tp \bar{F}\end{matrix}}&0\end{array}\right]},
	\end{align}
	where the notation is the same as in Theorem~\ref{thm:2}.
\end{cor}

\section{Interpreting the optimal controller}\label{sec:diss}
\paragraph{State interpretation.}
Based on Theorem \ref{thm:2}, a realization of the controller (calling the state $\xi$) is
\begin{subequations}\label{eq:state}
	\begin{align}
		\dot{\xi} &= (\bar{A}+\bar{L}\bar{C}+\bar{B}\bar{S}_m\bar{F}\bar{S}_n^{-1})\xi-\bar{L}\bar{\1}_p y\\
		u &= \bar{\1}_m^\tp \bar{F}\bar{S}_n^{-1}\xi.
	\end{align}
\end{subequations}
We call $\xi$ the \textit{state} coordinates. Informally, $\xi_i$ can be interpreted as agent $i$'s best estimate of the global state $x$ given the available information $y_{\bar i}$.
Likewise, a realization from Corollary \ref{corr} with state $\eta$ is
\begin{subequations}\label{eq:innov}
	\begin{align}
		\dot{\eta} &= (\bar{A}+\bar{S}_n^{-1}\bar{L}\bar{S}_p\bar{C}+\bar{B}\bar{F})\eta-\bar{S}_n^{-1}\bar{L}\bar{\1}_p y\\
		u &= \bar{\1}_m^\tp \bar{F} \eta.
	\end{align}
\end{subequations}
We call $\eta$ the \textit{innovation} coordinates. Informally, $\eta$ captures improvements in state estimates as information is aggregated along the DAG.
A block-diagram of the global controller is shown in Figure~\ref{fig:global}.

\begin{figure*}
	\centering
	\includegraphics{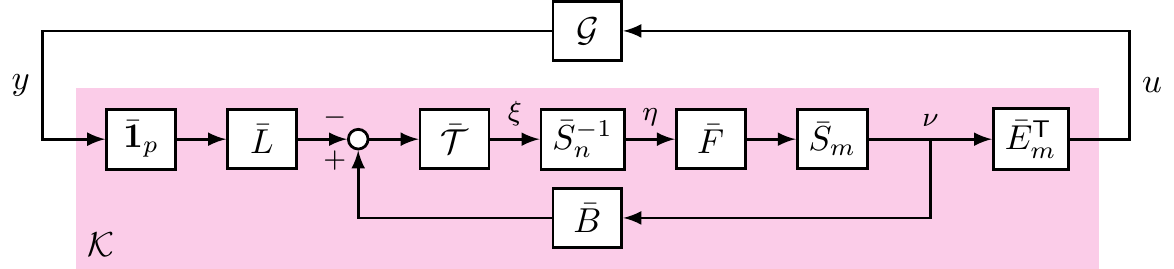}
	\caption{Block diagram representation of the global plant $\G$ and global controller $\K$ from Theorem \ref{thm:2}. Here, the (block-diagonal) estimator dynamics are contained in $\bar{\mathcal{T}} \defeq (sI-\bar A-\bar L\bar C)^{-1}$ and the remaining blocks of $\K$ are static.}\label{fig:global}
\end{figure*}

\paragraph{Relationship to posets.}\label{sec:poset}

The work by Shah and Parrilo~\cite{shah2011optimal,shah2013cal} considers a different version of this decentralized control problem but obtains a similar optimal structure. The authors consider agents communicating on a DAG, except all agents use \textit{state feedback} rather than output feedback, and $A$ is not required to be block-diagonal; it may have a block-sparsity pattern conforming to that of the DAG. The resulting controller for this problem has a controller-estimator structure very similar to the one in Figure~\ref{fig:global}.
What the state-feedback coupled case and the output-feedback decoupled case share in common is the notion of \textit{separability}~\cite{kim2015explicit}. The more general output-feedback case with dynamic coupling is \textit{not} separable and the control and estimation gains are generally coupled in that case~\cite{lessard2015optimal}.
Nevertheless, the fact that a state-feedback and output-feedback variants of the problem share a similar optimal controller structure suggests that there may exist an even broader class of control problems sharing this interesting structure.

\section{Agent-level implementation}\label{sec:impl}

In this section, we further refine the global controller structure of Figure~\ref{fig:global} to obtain agent-level implementations of the optimal controller.

The realizations~\eqref{eq:state} and~\eqref{eq:innov} describe the optimal controller dynamics in terms of the global state $\xi$ or $\eta$. Define the partial input $\nu \defeq \bar S_m \bar F \bar S_n^{-1} \xi$, and further isolate the dynamics of each individual agent by writing the $i^\text{th}$ block-component of~\eqref{eq:state} as
\begin{subequations}\label{eq:state_agent}
	\begin{align}
		\label{eq:state_agent_a}
		\dot \xi_i &= (A+ L_i C_2)\xi_i + B_2 \nu_i- L_i y \\
		\label{eq:state_agent_b}
		u_i &= E_{m_i}^{\tp} \nu_i = \nu_{i,i} \\
		\nu_i &= \sum_{j\in \bar{i}} F_j(\bar S_n^{-1} \xi)_j.\label{eq:state_agent_c}
	\end{align}
\end{subequations}
Note that~\eqref{eq:state_agent} can be implemented by agent $i$ because updates only require information from ancestors $\bar i$. Specifically, each $\nu_i$ is a function of $\xi_{\bar i}$ alone. Moreover, \eqref{eq:state_agent_a} does not depend on the full $y$, but rather only $y_{\bar i}$, since $L_i\defeq E_{n_{\bar{i}}}L^{\bar{i}}E_{p_{\bar{i}}}^\tp$ (defined in Theorem~\ref{thm:2}).

We will now subdivide~\eqref{eq:state_agent} to obtain agent-level update equations for each of the agents. It turns out agent $i$ only needs to store $\xi_{i,\underline i}$, which are the components of $\xi_i$ corresponding to the descendants of $i$.

By isolating these components, \eqref{eq:state_agent_a} reduces to
\begin{subequations}\label{agent:all}
	\begin{align}\label{agent:1}
		\dot \xi_{i,\underline{i}} &= A_{\underline{ii}} \xi_{i,\underline{i}} + B_{2_{\underline{ii}}} \nu_{i,\underline{i}}- E_{n_{\underline{i}}}^\tp E_{n_{{i}}} L^i (y_i - C_{2_{ii}} \xi_{i,i}).
	\end{align}
	Express~\eqref{eq:state_agent_c} in a more convenient form by expanding:
	$
	\nu_i = \sum_{j\in\bar i} F_j (\bar S_n^{-1}\xi)_j
	= \sum_{j\in\bar i} F_j \sum_{k \in\bar j} \bar S_n^{-1}(j,k) \xi_k
	$. Now define the summand $\tilde \nu_j \defeq F_j \sum_{k \in\bar j} \bar S_n^{-1}(j,k) \xi_k$. Since $F_i \defeq E_{m_{\underline i}} F^i E_{n_{\underline i}}^\tp$ (defined in Theorem~\ref{thm:2}), we can write $\nu_i$ compactly as: $\nu_i = \sum_{k\in\bar i} E_{m_{\underline k}} \tilde\nu_{k,\underline k}$, which leads to the following expression for $\nu_{i,\underline i}$.
	\begin{equation}\label{agent:2}
	\nu_{i,\underline i} = \tilde \nu_{i,\underline i} + E_{m_{\underline i}}^\tp \sum_{ k \in \bar{\bar{i}}} E_{m_{\underline k}} \tilde\nu_{k,\underline k}.
	\end{equation}
	Furthermore, each $\tilde \nu_{i,\underline i}$ is computable from the information available to agent $i$. To see why, compute:
	\begin{multline}\label{agent:3}
		\tilde \nu_{i,\underline i}
		= E_{m_{\underline i}}^\tp F_i \sum_{k \in\bar i} \bar S_n^{-1}(i,k) \xi_k
		= F^i \sum_{k \in\bar i} \bar S_n^{-1}(i,k) \xi_{k,\underline i} \\
		= F^i \sum_{k \in\bar i} \bar S_n^{-1}(i,k) E_{n_{\underline i}}^\tp E_{n_{\underline k}} \xi_{k,\underline k}.
	\end{multline}
\end{subequations}
In summary, the optimal agent-level controller for agent $i$ maintains a state $\xi_{i,\underline i}$ and uses update equations~\eqref{agent:all} and \eqref{eq:state_agent_b}. Agent $i$ receives $\{\xi_{k,\underline{k}}\}_{k \in \bar{\bar{i}}}$ and $\{\tilde{\nu}_{k,\underline{k}}\}_{k \in \bar{\bar{i}}}$ from its strict ancestors $\bar{\bar i}$, and transmits $\xi_{i,\underline{i}}$ and $\tilde{\nu}_{i,\underline{i}}$ to its strict descendants $\underline{\underline i}$. A block diagram of the agent-level structure is shown in Figure~\ref{fig:local}.

\begin{figure*}[htb]
	\centering
	\raisebox{5mm}{\includegraphics{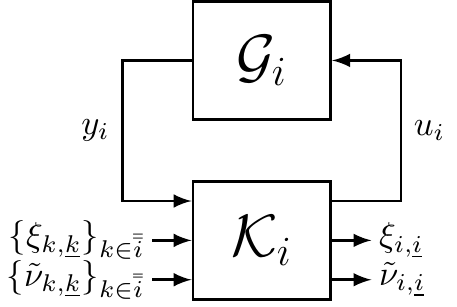}}
	\hfill
	\includegraphics{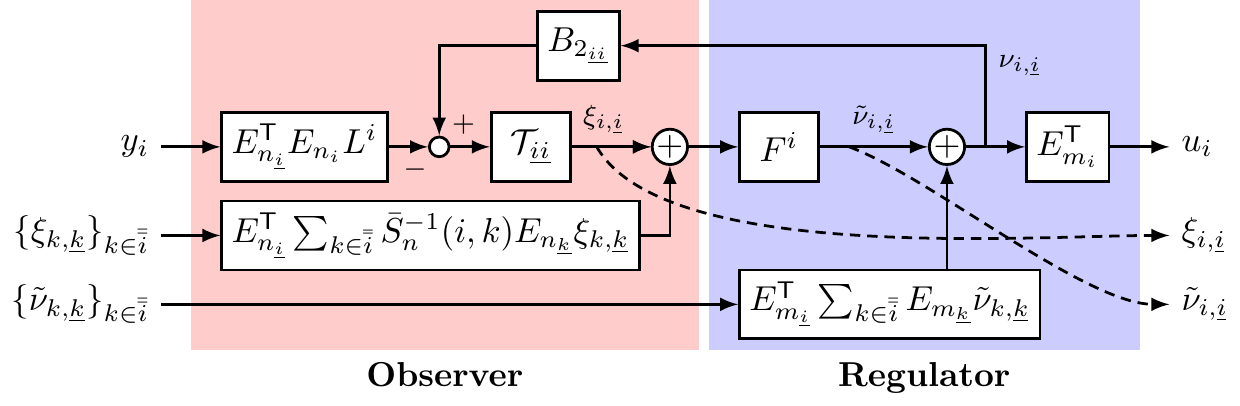}
	\vspace{-2mm}
	\caption{Block diagram representation of a local plant $\G_i$ and associated controller $\K_i$ (left) and the Observer-Regulator structure of $\K_i$ (right) from \eqref{agent:all}.
		Agent $i$ updates its local state $\xi_{i,\underline i}$ and computes its local input $u_i$ using the local measurement $y_i$ and the states $\xi_{k,\underline k}$ and partial inputs $\tilde v_{k,\underline k}$ from  strict ancestors $k\in \bar{\bar{i}}$.
		For simplicity, we used the following notation for the estimator dynamics: $\mathcal{T}_{\underline{ii}} \defeq  (sI-A_{\underline{ii}}-E_{n_{\underline{i}}}^{\tp} E_{n_i} L^i C_{2_{ii}} E_{n_{i}}^{\tp} E_{n_{\underline{i}}})^{-1}$.}
	\label{fig:local} 
\end{figure*}

A similar structure can be derived for the coordinate choice of Corollary~\ref{corr}, which we omit.

\paragraph{Minimal realization.} \label{sec:minimal}
As pointed out in Remark~\ref{rem:minimality} and above, the realizations obtained in Theorem~\ref{thm:2} and Corollary~\ref{corr} are not minimal.
The agent-level implementation encapsulates that each node $i$ doesn't transmit information about the $[N]\setminus \underline{i}$ nodes. These correspond to the unobservable modes that can be removed. Therefore, the reduced realizations are:
\begin{align*}
	\K&=\left[\begin{array}{c|c}%
		{\bar{E}_{\underline{n}}}^\tp (\bar{A}+\bar{L}\bar{C}+\bar{B}\bar{S}_m\bar{F}\bar{S}_n^{-1}){\bar{E}_{\underline{n}}}
		&-{\bar{E}_{\underline{n}}}^\tp\bar{L}\bar{\1}_p\\ \hlinet
		\bar{\1}_m^\tp\bar{F}\bar{S}_n^{-1}{\bar{E}_{\underline{n}}}&0
	\end{array}\right] \\
	\K&=\left[\begin{array}{c|c}
		{\bar{E}_{\underline{n}}}^\tp (\bar{A}+\bar{S}_n^{-1}\bar{L}\bar{S}_p\bar{C}+\bar{B}\bar{F}){\bar{E}_{\underline{n}}} & -{\bar{E}_{\underline{n}}}^\tp \bar{S}_n^{-1}\bar{L}\bar{\1}_p\\\hlinet
		\bar{\1}_m^\tp \bar{F}{\bar{E}_{\underline{n}}} &0
	\end{array}\right],
\end{align*}
where ${\bar{E}_{\underline{n}}} = \blkdiag(\{E_{n_{\underline{i}}}\})$.
In both realizations above, the total number of states is  $\sum_{i=1}^{N}n_{\underline{i}}$ and the state transition matrix has block-sparsity conforming to $S$.  

\section{Conclusions}
\label{sec:conclufuture}
In this work, our starting point was the explicit solution to the output-feedback dynamically decoupled cooperative control problem found by Kim and Lall~\cite{kim2015explicit}. We obtained an equivalent and compact realization of the optimal controller, which enabled us to derive agent-level implementations for the optimal control laws and to characterize precisely what each agent needs to receive, compute, and transmit.

\appendix

\section{Proof of Theorem~\ref{thm:2}}\label{sec:appendix}

In order to manipulate the expression for $\K$ in Lemma~\ref{lem:1}, we begin by finding a realization for $\Pp$ by aggregating the individual $\Pp_{\underline{i}i}$ given in~\eqref{eq:Pdef}. For ease of exposition, we augment the states of each $\Pp_{\underline{i}i}$ by adding modes that are simultaneously uncontrollable and unobservable. Specifically, we augment each $A_{\underline{ii}}$ to $A$ and likewise with $B_{2_{\underline{ii}}}$ to $B_2$ and $C_{2_{i\underline{i}}}$ to $C_2$.
For $i \in [N]$, we also augment the $F^i$ by zero-padding to obtain $F_i \defeq E_{m_{\underline{i}}}F^i E_{n_{\underline{i}}}^\tp$. The estimation gains are similarly zero-padded to obtain $\tilde{L} \defeq \blkdiag(\{E_{n_i}L^i E_{p_i}^\tp\})$. So finally, we obtain
\begin{equation}
\begin{aligned}
\Pp={\left[\begin{array}{c|c}%
	\bar{A}+\bar{B}\bar{F}+\tilde{L}\bar{C}& -\tilde{L}\bar{\1}_p\\\hlinet \bar{\1}_m^\tp \bar{F}&0\end{array}\right]}.
\end{aligned}
\end{equation}
The input and output matrices, $-\tilde{L}\bar{\1}_p$ and $\bar{\1}_m^\tp \bar{F}$ respectively, also arose by zero-padding so that all the modes we added to $\Pp$ are both uncontrollable and unobservable.
Now consider the related transfer matrix
\begin{equation}
\begin{aligned}
\Pp_2\defeq{\left[\begin{array}{c|c}%
	\bar{A}+\bar{B}\bar{F}+{L}_x\bar{C}& -\tilde{L}\bar{\1}_p\\\hlinet \bar{\1}_m^\tp \bar{F}&0\end{array}\right]},
\end{aligned}
\end{equation}
where $L_x\defeq\bar{S}_n^{-1}\bar{L}(\bar{S}_p-\bar{\1}_p(\bar{\1}_{p}^\tp -\diag(\bar{\1}_{p}^\tp )))$. We will now prove that $\Pp(s) = \Pp_2(s)$, which can be done by explicitly computing the transfer matrices and performing appropriate algebraic manipulations. Beginning with $\Pp_2$,
\begin{multline}
	\Pp_2(s) = -\bar{\1}_m^\tp \bar{F}\bigl( s\bar{I}_n-\bar{A}-\bar{B}\bar{F}-\bar{S}_n^{-1}\bar{L}\bar{S}_p\bar{C}\\
	+\bar{S}_n^{-1}\bar{L}\bar{\1}_p\bar{\1}_{p}^\tp \bar{C}-\bar{S}_n^{-1}\bar{L}\bar{\1}_p\diag(\bar{\1}_{p}^\tp )\bar{C}\bigr)^{-1}\tilde{L}\bar{\1}_p,
\end{multline}
where $\bar{I}_n \defeq I_N\otimes I_n = I_{Nn}$.
By the sparsity pattern of $\tilde L$ and $\bar{S}_n$, we can verify that $\bar{S}_n\tilde{L}\bar{\1}_p=\bar{L}\bar{\1}_p$. Therefore,
\begin{multline*}
	\Pp_2(s) = -\bar{\1}_m^\tp \bar{F}\bigl( s\bar{I}_n-\bar{A}-\bar{B}\bar{F}-\bar{S}_n^{-1}\bar{L}\bar{S}_p\bar{C}\\+\tilde{L}\bar{\1}_p\bar{\1}_{p}^\tp \bar{C}-\tilde{L}\bar{\1}_p\diag(\bar{\1}_{p}^\tp )\bar{C}\bigr)^{-1}\tilde{L}\bar{\1}_p.
\end{multline*}
Now, $\tilde{L}\bar{\1}_p\diag({\bar{\1}_{p}^{\tp}})\bar{C} = \tilde{L}\bar{C}$. Defining $Z_1 \defeq (s\bar{I}_n-\bar{A}-\bar{B}\bar{F}-\tilde{L}\bar{C})$ and $Z_2 \defeq -\bar{S}_n^{-1}\bar{L}\bar{S}_p\bar{C}+\tilde{L}\bar{\1}_p\bar{\1}_{p}^\tp \bar{C}$,
\[
\Pp_2(s) = -\bar{\1}_m^\tp \bar{F}\left( Z_1+Z_2 \right)^{-1}\tilde{L}\bar{\1}_p.
\]
Applying the Woodbury matrix identity, 
\begin{equation}\label{eq:Meq}
\Pp_2(s) = -\bar{\1}_m^\tp \bar{F}Z_1^{-1}\tilde{L}\bar{\1}_p+M=\Pp(s)+\mathcal{M}(s),
\end{equation}
where $\mathcal{M}(s)\defeq\bar{\1}_m^\tp \bar{F}Z_1^{-1}Z_2(\bar{I}_n+Z_1^{-1}Z_2)^{-1}Z_1^{-1}\tilde{L}\bar{\1}_p$.
The matrix $Z_1^{-1}$ is block-diagonal, where each $Z_{1_{ii}}$ is also block-diagonal. Partitioning the rows and columns according to observable and unobservable modes, respectively, we obtain the following block-sparsity patterns: $Z_1^{-1} \sim \sbmat{\star & 0 \\0 & \star}$, $Z_2 \sim \sbmat{0 & \star \\ 0 & \star}$, $\tilde{L}\bar{\1}_p \sim\sbmat{\star \\ 0}$, and $\bar{\1}_m^\tp \bar{F}\sim\sbmat{\star & 0}$. Substituting into the definition of $\mathcal{M}$, we conclude that $\mathcal{M}(s)=0$ and so from~\eqref{eq:Meq}, we have $\Pp_2(s)=\Pp(s)$.

Next, we substitute this expression for $\Pp_2$ into the expression for $\K$ in Lemma~\ref{lem:1}, one term at a time.
\[
\tilde{\Pp}\defeq\Pp-\diag(\Pp)=\left[\begin{array}{c|c}%
\bar{A}+\bar{B}\bar{F}+{L}_x\bar{C}& -\tilde{L}\bar{\1}_p\\
\hlinet F_\Delta &0
\end{array}\right],
\]
where $F_\Delta\defeq\bar{\1}_m^\tp \bar{F}-\diag(\bar{\1}_m^\tp \bar{F}) = (\bar{\1}_m^\tp -\diag(\bar{\1}_m^\tp ))\bar{F}$. 
We evaluate, $(I+\G\Pp - \diag(\G\Pp))^{-1} = (I+\G\tilde{\Pp})^{-1}=$
\[
\left[\begin{array}{cc|c}%
A&B_2F_\Delta&0\\\tilde{L}\bar{\1}_pC_2&\bar{A}+\bar{B}\bar{F}+{L}_x\bar{C}&-\tilde{L}\bar{\1}_p\\\hlinet -C_2&0&I\end{array}\right].
\]
Finally, $\K = \Pp(I+\G\tilde{\Pp})^{-1}=$
\[
\squeezemat{1pt}
\left[\begin{array}{ccc|c}%
\bar{A}+\bar{B}\bar{F}+{L}_x\bar{C}&\tilde{L}\bar{\1}_pC_2&0&-\tilde{L}\bar{\1}_p\\0&A&B_2F_\Delta&0\\0&\tilde{L}\bar{\1}_pC_2&\bar{A}+\bar{B}\bar{F}+{L}_x\bar{C}&-\tilde{L}\bar{\1}_p\\\hlinet \bar{\1}_m^\tp \bar{F} &0& 0 &0\end{array}\right]\!.
\]
Applying the similarity transform with $T_1 \defeq \sbmat{ \bar{I}_{n}& 0 & \bar{I}_{n}\\ 0 & I_{n} & 0\\0& 0 & \bar{I}_{n}}$ reveals an uncontrollable mode, which we eliminate:
\[
\K=\left[\begin{array}{cc|c}%
A&B_2F_\Delta&0\\\tilde{L}\bar{\1}_pC_2&\bar{A}+\bar{B}\bar{F}+{L}_x\bar{C}&-\tilde{L}\bar{\1}_p\\\hlinet 0& \bar{\1}_m^\tp \bar{F}& 0\end{array}\right]\!.
\]
Applying the similarity transform $T_2\defeq \sbmat{ I_{n}& 0\\ 0 & \bar{S}_n^{-1}}$ and using the relationships 
$\tilde{L}\bar{\1}_pC_2 = \tilde{L}\bar{C}\bar{\1}_n$, $\bar{S}_n\bar{A}=\bar{A}\bar{S}_n$, $\bar{S}_n\bar{B}=\bar{B}\bar{S}_m$, and  $\bar{C}\bar{S}_n^{-1}=\bar{S}_p^{-1}\bar{C}$, we obtain
\[
\K\!=\!
\squeezemat{2.1pt}
\left[\begin{array}{cc|c}%
A&B_2(\bar{\1}_m^\tp \!-\!\diag(\bar{\1}_m^\tp ))\bar{F}\bar{S}_n^{-1}&0\\\bar{S}_n\tilde{L}\bar{C}\bar{\1}_n&\bar{A}\!+\!\bar{B}\bar{S}_m\bar{F}\bar{S}_n^{-1}\!+\!\bar{S}_n{L}_x\bar{S}_p^{-1}\bar{C}& -\bar{S}_n\tilde{L}\bar{\1}_p\\\hlinet 0& \bar{\1}_m^\tp \bar{F}\bar{S}_n^{-1}&0\end{array}\right]\!.
\]
Finally, we apply  $T_3 \defeq \sbmat{ I_{n}& (\bar{\1}_n^\tp -\diag(\bar{\1}_n^\tp ))\bar{S}_n^{-1}\\ 0 & \bar{I}_{n}}$ and use the relationship
$\bar{S}_n\tilde{L}\bar{\1}_p = \bar{L}\bar{\1}_p$, eliminate the uncontrollable states, and obtain a reduced realization of the controller:
\begin{align}
\K ={\left[\begin{array}{c|c}%
	\begin{matrix}\bar{A}+\bar{L}\bar{C}+\bar{B}\bar{S}_m\bar{F}\bar{S}_n^{-1}\end{matrix}&\begin{matrix}-\bar{L}\bar{\1}_p\end{matrix}\\\hlinet {\begin{matrix}\bar{\1}_m^\tp \bar{F}\bar{S}_n^{-1}\end{matrix}}&0\end{array}\right]}.
		\tag*{\mbox{\raisebox{-5mm}{\qedhere}}}
\end{align}

\bibliographystyle{abbrv}
{\small\bibliography{optcont}}
	
\end{document}